\documentclass[10pt,a4paper]{article}
\usepackage{amssymb}
\usepackage{amsmath}
\usepackage{amsthm}
\usepackage[all]{xy}

\def\eps{\varepsilon}

\theoremstyle{plain}
	\newtheorem{Def}{Definition}[section]
	\newtheorem{Thm}{Theorem}[section]
	\newtheorem{Lemma}[Thm]{Lemma}
	\newtheorem{Cor}[Thm]{Corollary}
\theoremstyle{Def}

	\newtheorem{Ex}{Example}[section]

\begin{document}

\title{Assisted Problem Solving and Decompositions of Finite Automata\thanks{This work was supported in part by the grant VEGA 1/3106/06.}}
\author{Peter Ga\v zi \and Branislav Rovan
}

\date{{\small Department of Computer Science, Comenius University\\
Mlynsk\' a dolina, 842 48, Bratislava, Slovakia\\
\texttt{\{gazi,rovan\}@dcs.fmph.uniba.sk}\\
}
}

\maketitle              

\begin{abstract}
A study of assisted problem solving formalized via decompositions of
deterministic finite automata is initiated. The landscape of new
types of decompositions of finite automata this study uncovered is
presented. Languages with various degrees of decomposability between
undecomposable and perfectly decomposable are shown to exist.
\end{abstract}

\section{Introduction}

In the present paper we initiate the study of {\em assisted problem
solving}. We intend to model and study situations, where solution to
the problem can be sought based on some additional a priori
information about the inputs. One can expect to obtain simpler
solution in such case. There are similar approaches known in the
literature, most notably the notions of advice functions
\cite{BalcazarDiazGabarro}, where the additional information is
based on the length of the input word and the notion of promise
problems \cite{EvenSelmanYacobi}, where the set of inputs is
separated into three classes -- those with ``yes'' answer, those
with ``no'' answer and those where we do not care about the outcome.
By considering the simplest case where the ``problem solving''
machinery is the deterministic finite automaton (DFA) we obtain a
new motivation for studying new types of finite automata
decompositions.

In this paper we shall thus consider the case where solving a
problem shall mean constructing an automaton for a given language
$L$. The ``assistance'' shall be given by additional information
about the input, e.g., that we can assume the inputs shall be
restricted to words from a particular regular language $L'$. Thus,
instead of looking for an automaton $A$ such that $L=L(A)$ we can
look for a (possibly simpler) automaton $B$ such that $L=L(B) \cap
L'$. We can then say that $B$ accepts $L$ with the assistance of
$L'$. We shall call $L'$ (or the corresponding automaton $A'$ such
that $L'=L(A')$) an {\em advisor} to $B$. In this case the advisor
$A'$ provides assistance to the solver B by guaranteeing that $A'$
accepts the given input word. We shall also study a case where the
assistance provides more detailed information about the outcome of
the computation of $A'$ on the input word (e.g., the state reached).
Clearly the advisor can be considered useful only if it enables $B$
to be simpler than $A$ and at the same time $A'$ is not more
complicated than $A$. The measure of complexity we shall consider is
the number of states of the deterministic finite automaton. This
measure of complexity was used quite often recently due to renewed
interest in finite automata prompted by applications such as model
checking (see e.g. \cite{ShengYu} for a recent survey). (Note
that results complementary to ours, namely results on complexity of
automata for the intersection of regular sets were studied in
\cite{Birget}.)

The contribution of our paper is twofold. First, we can interpret
the `solver' and the `advisor' as two parallel processes each
performing a {\em different} task and jointly solving a problem.
Since our approach lends itself to a generalisation to $k$ advisors
it may stimulate new parallel solutions to problems (the traditional
ones usually using parallel processes to perform essentially the
same task). Second, the choice of finite automata as the simplest
problem solving machinery brought about new types of decompositions
motivated by the information the `advisor' can provide to the
`solver'. Our results provide a complete picture of the landscape of
these decompositions.

The problem within this scenario we shall address in this paper is the existence of a
useful advisor for a given automaton $A$. We shall compare the power of several
types of advisors, and investigate the effect of the advisor on the
complexity of the assisted solver $B$. We can formulate this also as
a problem of decomposition of deterministic finite state automata --
given DFA $A$ find DFA $A_1$ (a solver) and  $A_2$ (an advisor)
such that $w \in L(A)$ can be determined from the computations of $A_1$
and $A_2$. We shall study several new types of decompositions of
DFA, one of them is analogous to the state behavior decomposition of
finite state transducers studied in \cite{HartmanisStearns}. In
Sect.~\ref{sec:relations} we prove relations among these decompositions. For each
type of decomposition there are automata which are undecomposable
and automata for which there is a decomposition that is the best
possible. In Sect.~\ref{sec:degrees} we consider the space between these extreme
points and study the degree of decomposability.

\section{Definitions and Notation}\label{sec:def}

We shall use standard notions of the theory of formal languages (see
e.g. \cite{HopcroftUllman}). Our notation shall be as follows.
$\Sigma^*$ denotes the set of all words over the alphabet $\Sigma$,
the length of a word $w$ is denoted by $|w|$, $\eps$ denotes the empty
word, and for a language $L$ we shall denote by $\Sigma_L$ the
minimal alphabet such that $L\subseteq\Sigma_L^*$. The number of
occurrences of a given letter $a$ in a word $w$ is denoted by
$\#_a(w)$. Throughout this paper we shall consider deterministic
finite automata only.

A \emph{deterministic finite automaton} (DFA) is a quintuple
$(K,\Sigma,\delta,q_0,F)$, such that $K$ is a finite set of states,
$\Sigma$ is a finite input alphabet, $q_0\in K$ is the initial
state, $F\subseteq K$ is the set of accepting states and
$\delta\colon K\times\Sigma\to K$ is a transition function. As
usual, we shall denote by $\delta$ also the standard extension of
$\delta$ to words, i.e., $\delta\colon K\times\Sigma^*\to K$. We
shall denote by $|K|$ the number of states in $K$.

Formalizing the notions of assisted problem solving from the
Introduction we shall now define several types of decompositions of
DFA $A$ into two (simpler) DFAs $A_1$ and $A_2$ (a solver and an
advisor) so that the membership of an input word $w$ in $L(A)$ can be
determined based on the information on the computations of $A_1$ and
$A_2$ on $w$.

We first introduce an \emph{acceptance-identifying} decomposition of
deterministic finite automata.

\begin{Def}
A pair of DFAs $(A_1,A_2)$, where
$A_1=(K_1,\Sigma,\delta_1,q_1,F_1)$ and $A_2=(K_2,\Sigma,\delta_2,q_2,F_2)$, forms an
{\em acceptance-identifying decomposition} (AI-decomposition) of a DFA $A=(K,\Sigma,\delta,q_0,F)$, if
$L(A)=L(A_1)\cap L(A_2)$.
This decomposition is {\em nontrivial} if $|K_1|<|K|$ and $|K_2|<|K|$.
\end{Def}

By decomposing $A$ in this manner, one of the decomposed automata (say $A_2$) can act as
an advisor and narrow down the set of input words for the other one (say $A_1$), whose
task to recognize the words of $L(A)$ may become easier.

Another requirement we could pose on a decomposition is to identify the final state of any
computation of the original automaton by only knowing the final states of both
corresponding
computations of the automata forming the decomposition. This requirement can be formalized
as follows.

\begin{Def}\label{def:SI}
A pair of DFAs $(A_1,A_2)$, where
$A_1=(K_1,\Sigma,\delta_1,q_1,F_1)$ and $A_2=(K_2,\Sigma,\delta_2,q_2,F_2)$, forms a
{\em state-identifying
decomposition} (SI-decomposition) of a DFA $A=(K,\Sigma,\delta,q_0,F)$, if there exists a mapping
$\beta\colon K_1\times K_2\to K$, such that it holds
$\beta(\delta_1(q_1,w),\delta_2(q_2,w))=\delta(q_0,w)$ for all $w\in\Sigma^*$.
This decomposition is {\em nontrivial} if $|K_1|<|K|$ and $|K_2|<|K|$.
\end{Def}

The third -- and the weakest -- requirement we pose on a
decomposition of a DFA is to require that there must exist a way to
determine whether the original automaton would accept some given
input word based on knowing the states in which the computations of
both decomposition automata have finished.

\begin{Def}\label{def:wAI}
A pair of DFAs $(A_1,A_2)$, where
$A_1=(K_1,\Sigma,\delta_1,q_1,F_1)$ and $A_2=(K_2,\Sigma,\delta_2,q_2,F_2)$, forms a
weak {\em acceptance-identifying decomposition} (wAI-decomposition) of a DFA $A=(K,\Sigma,\delta,q_0,F)$,
if there exists a relation $R\subseteq K_1\times K_2$ such that it holds
$R(\delta_1(q_1,w),\delta_2(q_2,w))\Leftrightarrow w\in L(A)$ for all $w\in\Sigma^*$.
This decomposition is {\em nontrivial} if $|K_1|<|K|$ and $|K_2|<|K|$.
\end{Def}

Note that in the last two definitions, the sets of accepting states of $A_1$ and $A_2$ are
irrelevant.

By a decomposability of a regular language $L$ in some
way, we shall mean the decomposability of the corresponding minimal
automaton over $\Sigma_L$.

To be able to compare these new types of decomposition to the
\emph{parallel decompositions of state behavior} introduced for
sequential machines in \cite{HartmanisStearns}, we shall redefine
them for DFAs.

\begin{Def}\label{def:realization}
A DFA $A'=(K',\Sigma,\delta',q_0',F')$ is said to \emph{realize the state behavior} of a DFA
$A=(K,\Sigma,\delta,q_0,F)$ if there exists an injective mapping $\alpha\colon K\to K'$
such that
\begin{enumerate}
\item[(i)]
$(\forall a\in\Sigma)(\forall q\in K); \delta'(\alpha(q),a)=\alpha(\delta(q,a))$,
\item[(ii)]
$\alpha(q_0)=q_0'$.
\end{enumerate}
Moreover, $A'$ is said to \emph{realize the state and acceptance behavior} of $A$, if
in addition the following property holds:
\begin{enumerate}
\item[(iii)]
$(\forall q\in K); \alpha(q)\in F'\Leftrightarrow q\in F$.
\end{enumerate}
\end{Def}

\begin{Def}
The \emph{parallel connection} of two DFA $A_1=(K_1,\Sigma,\delta_1,q_1,F_1)$ and
$A_2=(K_2,\Sigma,\delta_2,q_2,F_2)$ is the DFA
$A=A_1||A_2=(K_1\times K_2,\Sigma,\delta,(q_1,q_2),$ $F_1\times F_2)$ such
that $\delta((p_1,p_2),a)=(\delta_1(p_1,a),\delta_2(p_2,a))$.
\end{Def}

\begin{Def}\label{def:decomposition}
A pair of DFAs $(A_1,A_2)$ is a \emph{state behavior (SB-)
decomposition} of a DFA $A$ if
$A_1||A_2$ realizes the state behavior of $A$. The pair $(A_1,A_2)$ is an \emph{acceptance
and state behavior (ASB-)
decomposition} of $A$ if $A_1||A_2$ realizes the state and acceptance behavior of $A$.
This decomposition is \emph{nontrivial} if both $A_1$ and $A_2$ have fewer states than $A$.
\end{Def}

We have
modified the definitions to fit the formalism and purpose of deterministic finite
automata (i.e., to accept formal languages) without loosing the connection to the
strongly related and useful concept of \emph{S.P.partitions}, exhibited below.

We shall use the following notation and properties of S.P. partitions from
\cite{HartmanisStearns}. A partition $\pi$ on a set of states of a DFA
$A=(K,\Sigma,\delta,q_0,F)$  has \emph{substitution property}
(S.P.), if it holds $\forall p,q\in K;~~ p\equiv_\pi q \Rightarrow (\forall a\in\Sigma;
\delta(p,a)\equiv_\pi\delta(q,a))$.
If $\pi_1$ and $\pi_2$ are partitions on a given set $M$, then
\begin{enumerate}
\item[(i)]
$\pi_1\cdot\pi_2$ is a partition on $M$ such that
$a\equiv_{\pi_1\cdot\pi_2}b\Leftrightarrow a\equiv_{\pi_1}b\land a\equiv_{\pi_2}b$,
\item[(ii)]
$\pi_1+\pi_2$ is a partition on $M$ such that
$a\equiv_{\pi_1+\pi_2}b$ iff there exists a sequence $a=a_0, a_1, a_2, \ldots,
a_n=b$, such that $a_i\equiv_{\pi_1}a_{i+1}\lor a_i\equiv_{\pi_2}a_{i+1}$ for all
$i\in\{0,\ldots,n-1\}$,
\item[(iii)]
$\pi_1\preceq\pi_2$ if it holds $(\forall x,y\in M);~~ x\equiv_{\pi_1}y\Rightarrow
x\equiv_{\pi_2}y$.
\end{enumerate}
The set of all partitions on a given set (with the partial order $\preceq$,
join realized by $+$ and meet realized by $.$) forms a
lattice. The set of all S.P. partitions on the set of states of a given DFA
forms a sublattice of the lattice of all partitions on this set.
The trivial partitions $\{\{q_0\},
\{q_1\},\ldots ,\{q_n\}\}$ and $\{\{q_0,q_1,\ldots,q_n\}\}$ shall be denoted by symbols $0$ and $1$,
respectively. The block of a partition $\pi$ containing the state $q$ shall be denoted by
$[q]_{\pi}$.
In addition, we shall use the following separation notion.

\begin{Def}\label{def:separate}
The partitions $\pi_1=\{R_1,\ldots ,R_k\}$ and $\pi_2=\{S_1,\ldots ,S_l\}$ on
a set of states of a DFA $A=(K,\Sigma,\delta,q_0,F)$ are said to \emph{separate the final
states}
of $A$ if there exist indices $i_1,\ldots,i_r$ and $j_1,\ldots,j_s$ such that it holds
$(R_{i_1}\cup\ldots\cup R_{i_r})\cap(S_{j_1}\cup\ldots\cup S_{j_s})=F$.
\end{Def}

\section{Relations Between Types of Decompositions}\label{sec:relations}

The concept of partitions separating the final states allows us to derive a
necessary and sufficient condition for the existence of SB- and ASB-decompositions similar to the
one stated in \cite{HartmanisStearns}.

\begin{Thm}\label{thm:SB-condition}
A DFA $A=(K,\Sigma,\delta,q_0,F)$ has a nontrivial
SB-decomposition iff there exist two nontrivial S.P.
partitions $\pi_1$ and $\pi_2$ on the set of states of $A$ such that $\pi_1\cdot\pi_2=0$.
This decomposition is an ASB-decomposition if and only if these partitions separate the final states of $A$.
\end{Thm}

\begin{proof}
The proof is analogous to that in \cite{HartmanisStearns} but had to be
extended for the ASB-decomposition. We omit it due to space constraints.
\end{proof}

For the other decompositions, we can derive the following sufficient conditions that
exploit the concept of S.P. partitions.

\begin{Thm}\label{thm:AI-condition}
Let $A=(K,\Sigma,\delta,q_0,F)$ be a deterministic finite automaton, let $\pi_1$ and
$\pi_2$ be nontrivial S.P. partitions on the set of states of $A$, such that
they separate the final states of $A$. Then $A$ has a nontrivial AI-decomposition.
\end{Thm}

\begin{proof}
Since $\pi_1$ and $\pi_2$ separate the final states of $A$, there exist blocks $B_1,
\ldots, B_k$ and $C_1, \ldots, C_l$ of the partitions $\pi_1$ and $\pi_2$ respectively, such that $(B_1\cup\ldots\cup
B_k)\cap(C_1\cup\ldots\cup C_l)=F$.
We shall construct two automata $A_1$
and $A_2$ having states corresponding to blocks of these partitions and show that
$(A_1,A_2)$ is a nontrivial AI-decomposition of $A$. Let
$A_1=(\pi_1,\Sigma,\delta_1,[q_0]_{\pi_1},\{B_1, \ldots, B_k\})$ and
$A_2=(\pi_2,\Sigma,\delta_2,[q_0]_{\pi_2},\{C_1, \ldots, C_l\})$
be DFAs with
$\delta_i$ defined by $\delta_i([q]_{\pi_i},a)=[\delta(q,a)]_{\pi_i}$,
$i\in\{1,2\}$ (this definition does not depend on the choice of $q$ since $\pi_i$ is an S.P.
partition). We now need to prove that $L(A)=L(A_1)\cap L(A_2)$.

Let $w\in L(A)$. Suppose that the computation of $A$ on the word $w$ ends in some accepting state
$q_f\in F$. Then, from the construction of $A_1$ and $A_2$ it follows that the computation
of $A_i$ on the word $w$ ends in the state corresponding to the block $[q_f]_{\pi_i}$ of
the partition
$\pi_i$. Since $q_f\in F$, it must hold $[q_f]_{\pi_1}\in\{B_1,\ldots,B_k\}$ and
$[q_f]_{\pi_2}\in\{C_1,\ldots,C_l\}$, hence from the construction of $A_i$, these blocks
correspond to the accepting states in the respective automata. Thus $w\in
L(A_i)$ for $i\in\{1,2\}$, therefore $L(A)\subseteq L(A_1)\cap L(A_2)$.

Now suppose $w\in L(A_1)\cap L(A_2)$, Thus the computation of $A_1$ on $w$ ends in one of the
states $B_1, \ldots, B_k$, which means that the computation of $A$ on $w$ would end in a state
from the union of blocks $B_1\cup\ldots\cup B_k$. Using the same argument for $A_2$, we
get that the computation of $A$ on $w$ would end in a state from $C_1\cup \ldots \cup C_l$.
Since $(B_1\cup\ldots\cup B_k)\cap(C_1\cup\ldots\cup C_l)=F$ we obtain that the
computation of $A$ ends in an accepting state, hence $w\in L(A)$ and $L(A_1)\cap L(A_2)\subseteq L(A)$.

Since both partitions are nontrivial, so is the AI-decomposition obtained.
\end{proof}

\begin{Thm}\label{thm:wAI-condition}
Let $A=(K,\Sigma,\delta,q_0,F)$ be a deterministic finite automaton, let $\pi_1$ and
$\pi_2$ be nontrivial S.P. partitions on the set of states of $A$, such that
\hbox{$\pi_1\cdot\pi_2\preceq\{F,K - F\}$}. Then $A$ has a nontrivial wAI-decomposition.
\end{Thm}

\begin{proof}
We shall construct $A_1$ and $A_2$ corresponding to the S.P. partitions $\pi_1$ and
$\pi_2$ as follows:
$A_i=(\pi_i,\Sigma,\delta_i,[q_0]_{\pi_i},\emptyset)$, where
\hbox{$\delta_i([q]_{\pi_i},a)=[\delta(q,a)]_{\pi_i}$} and \hbox{$i\in\{1,2\}$}. To show
that $(A_1,A_2)$ is a wAI-decomposition of $A$, we define the
relation $R\subseteq \pi_1\times\pi_2$ by the equivalence $R(D_1,D_2)\Leftrightarrow(D_1\cap
D_2\subseteq F)$,where $D_i$ is some block of the partition $\pi_i$.
Now we need to prove that $\forall w\in\Sigma^*$;
$w\in L(A)\Leftrightarrow
R(\delta_1([q_0]_{\pi_1},w),\delta_2([q_0]_{\pi_2},w))$.

Let the computation of $A$ on $w$ end in some state
$p\in K$. It follows that the computation of
$A_i$ on the word $w$ ends in the state corresponding to the block $[p]_{\pi_i}$,
$i\in\{1,2\}$. Thus
$R(\delta_1([q_0]_{\pi_1},w),\delta_2([q_0]_{\pi_2},w))
\Leftrightarrow R([p]_{\pi_1},[p]_{\pi_2})$ and by the definition of $R$, we have
$R(\delta_1([q_0]_{\pi_1},w),\delta_2([q_0]_{\pi_2},w))
\Leftrightarrow [p]_{\pi_1}\cap [p]_{\pi_2}\subseteq F$. Since
$p\in[p]_{\pi_1}\cap [p]_{\pi_2}$, $[p]_{\pi_1}\cap [p]_{\pi_2}$ is a block
of the partition $\pi_1\cdot\pi_2$ and $\pi_1\cdot\pi_2\preceq\{F,K - F\}$, it must hold
that either $[p]_{\pi_1}\cap [p]_{\pi_2}\subseteq F$ or $[p]_{\pi_1}\cap
[p]_{\pi_2}\subseteq K - F$. Therefore
$R(\delta_1([q_0]_{\pi_1},w),\delta_2([q_0]_{\pi_2},w))
\Leftrightarrow p\in F$ and the proof is complete.
\end{proof}

It follows directly from the definitions, that each SI-decom\-po\-sition is
also a wAI-decomposition, and so is each AI-decomposition. Also, each ASB-decomposition is
an AI-decomposition, which is a consequence of the definition of acceptance and state
behavior realization. For minimal automata, a relationship between AI- and SI-decompositions can be obtained.

\begin{Thm}\label{thm:AI-SI-min}
Let $A=(K,\Sigma,\delta,q_0,F)$ be a minimal DFA, let $(A_1,A_2)$ be its
AI-decomposition. Then $(A_1,A_2)$ is also an SI-decomposition of $A$.
\end{Thm}

\begin{proof}
Since $(A_1,A_2)$ is an AI-decomposition of $A$, $L(A)=L(A_1)\cap L(A_2)$. Therefore if we
use the well-known Cartesian product
construction, we obtain the automaton $A_1||A_2$ such that $L(A_1||A_2)=L(A)$.
Since $A$ is the minimal automaton accepting the language $L(A)$, there exists a mapping
$\beta\colon K'\to K$ such that it holds
$(\forall w\in\Sigma^*);~~
\beta(\delta'(q_0',w))=\delta(\beta(q_0'),w)$,
where $\delta'$ is the transition function of $A_1||A_2$, $K'$ is its
set of states  and $q_0'$ is its initial state. Since $A_1||A_2$ is a parallel
connection (i.e., $K'=K_1\times K_2$, $q_0'$ is the pair of initial states of $A_1$ and
$A_2$), it is easy to see that $\beta$ is in fact
exactly the mapping required by the definition of the SI-decomposition.
\end{proof}

The ASB-decomposition is a combination of the SB-decomposition and the AI-decomposition,
as the next theorem shows.

\begin{Thm}\label{thm:ASB=AI+SB}
Let $A$ be a DFA without unreachable states. $(A_1,A_2)$ is an ASB-decomposition of $A$ iff $(A_1,A_2)$ is
both an SB-decomposition and an AI-decomposition of $A$.
\end{Thm}

\begin{proof}
The first implication clearly follows from the definitions, Theorem~\ref{thm:SB-condition} and
Theorem~\ref{thm:AI-condition}.
Now let $(A_1,A_2)$ be an SB- and AI-decomposition of $A=(K,\Sigma,\delta,q_0,F)$. Let $\alpha$ be the mapping
given by the definition of SB-decomposition. We need to prove that for all states $q$ of
$A$, $q\in F$ iff $\alpha(q)\in F_1\times F_2$, where $F_i$ is the set of
accepting states of $A_i$, $i\in\{1,2\}$. Let $q\in K$ and let $w$ be a word such that
$\delta(q_0,w)=q$. Then
$q\in F\Leftrightarrow w\in L(A)\Leftrightarrow w\in
L(A_1)\cap L(A_2)\Leftrightarrow\alpha(q)\in F_1\times F_2$,
where the first equivalence is implied by the choice of $w$, the second holds because
$(A_1,A_2)$ is an AI-decomposition and the third is a consequence of the properties of
$\alpha$ guaranteed by the SB-decomposition definition.
\end{proof}

There is also a relationship between SB- and SI-decompositions, in fact SB- is a stronger
version of the state-identifying decomposition, as the following two theorems show. We
need the notion of reachability on pairs of states.

\begin{Def}
Let $A_1=(K_1,\Sigma,\delta_1,p_1,F_1)$ and $A_2=(K_2,\Sigma,\delta_2,p_2,F_2)$ be DFAs.
We shall call a pair of states $(q,r)\in K_1\times K_2$ \emph{reachable}, if there exists
a word $w\in\Sigma^*$ such that $\delta_1(p_1,w)=q$ and
$\delta_2(p_2,w)=r$.
\end{Def}

\begin{Thm}\label{thm:SI-SB}
Let $A=(K,\Sigma,\delta,q_0,F)$ be a DFA and let $(A_1,A_2)$ be its
SB-decomposition. Then $(A_1,A_2)$ also forms an SI-decomposition of $A$.
\end{Thm}

\begin{proof}
Let
$A_i=(K_i,\Sigma,\delta_i,q_i,F_i)$, $i\in\{1,2\}$.
Since $(A_1,A_2)$ is an SB-decomposition of $A$, there exists an injective mapping $\alpha\colon K\to
K_1\times K_2$  such that it holds $\alpha(q_0)=(q_1,q_2)$ and
$(\forall a\in\Sigma)(\forall p\in K);
\alpha(\delta(p,a))=(\delta_1(p_1,a),\delta_2(p_2,a))$,
where $\alpha(p)=(p_1,p_2)$. Let us define a new mapping $\beta\colon K_1\times K_2\to K$
by
\begin{equation}\label{eq:def_beta}
\beta(p_1,p_2)=\left\{
\begin{array}{ll}
p & \textrm{if $\exists p\in K, \alpha(p)=(p_1,p_2)$} \\
q_0 & \textrm{otherwise.}
\end{array}
\right .
\end{equation}
Since $\alpha$ is injective, there exists at most one such $p$ and this definition is
correct.

We now need to prove that $\beta$ satisfies the condition from the definition of
SI-decomposition, i.e., that
$(\forall w\in\Sigma^*);~
\beta(\delta_1(q_1,w),\delta_2(q_2,w))=\delta(q_0,w)$.
Since $\alpha(q_0)=(q_1,q_2)$ and all the pairs of states we encounter in the computation
of $A_1||A_2$ are thus reachable, this follows from the definition of $\alpha$  and
(\ref{eq:def_beta}) by an easy induction.
\end{proof}

\begin{Lemma}\label{lemma:SI-SB}
Let $A$ be a DFA without unreachable states and let $(A_1,A_2)$ be its SI-decomposition, with
$\beta$ being the corresponding mapping. Then $(A_1,A_2)$ is an SB-decomposition of $A$ if and
only if $\beta$ is injective on all reachable pairs of states.
\end{Lemma}

\begin{proof}
Let $(A_1,A_2)$ be an SB-decomposition of $A$. It clearly follows from
Definition~\ref{def:SI}, that the corresponding $\beta$ satisfies the equation
(\ref{eq:def_beta}) in the proof of Theorem~\ref{thm:SI-SB} on all reachable pairs of
states. Since the mapping $\alpha$ is a bijection between the set of states of
$A$ and the set of all reachable pairs of states of $A_1$ and $A_2$, $\beta$ defined as
its inverse on the set of reachable pairs of states will be injective on this set.

For the other implication, let $(A_1,A_2)$ be an SI-decomposition of $A$ and let $\beta$
be injective on the set of
reachable pairs of states, let $\beta_r$ denote the mapping $\beta$ restricted onto the
set of all reachable pairs of states of $A_1,A_2$. Since $A$ has no unreachable states,
$\beta_r$ is also surjective, thus we can define a new mapping $\alpha\colon K\to
K_1\times K_2$ by the equation $\alpha(q)=\beta_r^{-1}(q)$. Since $\beta$ maps the initial
state onto the initial state, so does $\alpha$, and since $\beta$ satisfies the condition
from the Definition~\ref{def:SI}, it implies that also $\alpha$ satisfies the
condition (i) from the definition of realization of state behavior. Therefore  $(A_1,A_2)$
is an SB-decomposition of $A$, with the corresponding mapping $\alpha$.
\end{proof}

The converse of Theorem~\ref{thm:SI-SB} does not hold. The minimal automaton
for the language $L=\{a^{4k}b^{4l}|k\geq 0, l\geq 1\}$ gives a counterexample. Inspecting its
S.P. partitions shows that it has no nontrivial SB-decomposition, but it can be
AI-decomposed into minimal automata for languages $L_1=\{a^{4k}b^{l}|k\geq 0,l\geq 1\}$
and $L_2=\{w|\#_b(w)=4l; l\geq 0\}$.
According to Theorem~\ref{thm:AI-SI-min}, this AI-decomposition is also state-identifying.

Each ASB-decomposition is obviously also an SB-decomposition.
On the other hand, there exist SB-decomposable automata, that are
ASB-undecomposable. For example, the minimal automaton for the language
\begin{eqnarray*}
L_1&=&\{w\in\{a,b,c\}^*|\#_a(w) \mod 3=0\land \#_b(w) \mod 5=0\}\cup\\
&\cup&\{w\in\{a,b,c\}^*|\#_a(w) \mod 3=2\land \#_b(w) \mod 5=4\}
\end{eqnarray*}
has this
property, because the corresponding S.P. partitions on the set of its states
do not separate the final states in the sense of Definition~\ref{def:separate}.

It is also not so difficult to see that for any non-minimal automaton $A$ without unreachable
states, there exists a nontrivial AI- and
wAI-decomposition $(A_1,A_2)$ such that $A_1$ is the minimal automaton equivalent to $A$
and $A_2$ has only one state. This decomposition is obviously not state-identifying.

Figure~\ref{fig:relations} summarizes all the relationships
among the decomposition types that we have shown so far.

\begin{figure}
\begin{minipage}{0.4\textwidth}
\begin{displaymath}
\xymatrix{
    & ASB \ar[dl]\ar[dr]        \\
AI\ar@<1ex>[ur]|-{\times}\ar[dr]|-{min}\ar[dd]\ar@{.>}@<1ex>[dr]|-{\times}  & & SB\ar[dl]\ar@<1ex>[ul]|-{\times}    \\
    & SI\ar[dl]\ar@<1ex>[ur]|-{\times}              \\
 wAI\ar@{.>}@<1ex>[ur]|-{\times}                \\
}
\end{displaymath}
\end{minipage}
\begin{minipage}{0.6\textwidth}
Description:
\begin{description}
\item[$\xymatrix{A\ar[r] & B}$:]
every A-decomposition is also a B-decomposition
\item[$\xymatrix{A\ar@{.>}[r]|-{\times} & B}$:]
not every A-decomposition is also a B-decomposition
\item[$\xymatrix{A\ar[r]|-{\times} & B}$:]
there exists a DFA that has a nontrivial A-decomposition but does not have a nontrivial B-decomposition
\end{description}
\end{minipage}
\caption{Relationships between decomposition types of DFA}\label{fig:relations}
\end{figure}

Now we show that for the case of so-called perfect decompositions, some of the
types of decomposition mentioned coincide.

\begin{Def}
Let $t$ be a type of decomposition, $t\in\{ASB,SB,AI,SI,$ $wAI\}$. Let $A$ be a DFA having $n$
states, let $A_1$ and $A_2$ be DFAs having $k$ and $l$ states, respectively. We shall call
the pair $(A_1,A_2)$ a \emph{perfect $t$-decomposition} of $A$, if it forms a
$t$-decomposition of $A$ and $n=k\cdot l$.
\end{Def}

\begin{Thm}\label{thm:perfect-SI-SB}
Let $A$ be a DFA with no unreachable states and let $(A_1,A_2)$ be a pair of DFAs. Then
$(A_1,A_2)$ forms a perfect SI-decomposition of $A$ iff $(A_1,A_2)$ forms a
perfect SB-decomposition of $A$.
\end{Thm}

\begin{proof}
One of the implications is a consequence of Theorem~\ref{thm:SI-SB}. As to the second
one, since $(A_1,A_2)$ forms a perfect SI-decomposition of $A$, each of the pairs of
states of $A_1$ and $A_2$ is reachable and each pair has to correspond to a different
state of $A$ in the mapping $\beta$, therefore $\beta$ is bijective and the theorem
follows from Theorem~\ref{lemma:SI-SB}.
\end{proof}

\begin{Cor}
Let $A$ be a minimal DFA and let $(A_1,A_2)$ be a pair of DFAs. Then
$(A_1,A_2)$ forms a perfect AI-decomposition of $A$ iff $(A_1,A_2)$ forms a
perfect ASB-decomposition of $A$.
\end{Cor}

\begin{proof}
The claim follows from Theorem~\ref{thm:ASB=AI+SB}, Theorem~\ref{thm:AI-SI-min} and
Theorem~\ref{thm:perfect-SI-SB}.
\end{proof}

As a consequence of these facts, we can use the necessary and sufficient conditions stated in
Theorem~\ref{thm:SB-condition} to look for perfect AI- and SI-de\-com\-po\-si\-tions.

Now, let us inspect the relationship between decompositions of an automaton and the
decompositions of the corresponding minimal automaton.

\begin{Thm} \label{thm:minimal}
Let $A=(K,\Sigma,\delta,q_0,F)$ be a DFA and let $A_{\rm{min}}$ be a minimal DFA such that $L(A)=L(A_{\rm{min}})$. Let
$(A_1,A_2)$ be an SI-decomposition (AI-decomposition, wAI-decomposition) of $A$, then $(A_1,A_2)$ also forms a decomposition of $A_{\rm{min}}$ of the
same type.
\end{Thm}

\begin{proof}
First, note that this theorem does not state that any of the decompositions is nontrivial.
To prove the statement for SI-decompositions,  suppose that $(A_1,A_2)$ is an SI-decomposition of $A$, thus there exists a mapping
$\alpha\colon K_1\times K_2\to K$ such that it holds
$(\forall w\in\Sigma^*);~
\alpha(\delta_1(q_1,w),\delta_2(q_2,w))=\delta(q_0,w)$,
where $\delta_i$ and $q_i$ are the transition function and the initial state of the
automaton $A_i$. Since $A_{\rm{min}}$ is the minimal automaton corresponding to $A$, there
exists some mapping $\beta\colon K \to K_{\rm{min}}$ such that
$(\forall w\in\Sigma^*);~
\beta(\delta(q_0,w))=\delta_{\rm{min}}(\beta(q_0),w)$,
where $\delta_{\rm{min}}$ is the transition function of $A_{\rm{min}}$ and $K_{\rm{min}}$ is the
set of states of $A_{\rm{min}}$. By the composition of these mappings we obtain
the mapping $\beta\circ\alpha\colon K_1\times K_2\to K_{\rm{min}}$,
which combines $A_1$ and $A_2$ into
$A_{\rm{min}}$ in the way that the definition of SI-decomposition requires.
For both the AI- and the wAI-decomposition, this statement is trivial, since
$L(A)=L(A_{\rm{min}})$.
%
%
%
\end{proof}

Based on the above theorem it thus suffices to inspect the SI- \mbox{(AI-, wAI-)} decomposability of the minimal
automaton accepting a given language, and if we show its undecomposability, we know that
the recognition of this language cannot be simplified using an advisor of the respective
type. However, this does not hold for SB- and ASB-decompositions, as exhibited by the
following example.
\begin{Ex}\label{ex1}
Let us consider the language $L=\{a^{2k}b^{2l}|k\geq 0, l\geq 1\}$. The minimal automaton
$A_{\rm{min}}=(K,\Sigma_L,\delta,a_0,\{a_0,b_0\})$ has its transition function defined by the
first transition diagram in Fig.\ref{fig:example}.
We can easily show that this automaton does not have any nontrivial SB- (and thus neither
ASB-) decomposition by enumerating its S.P. partitions.

\begin{figure}
\begin{minipage}{0.5\textwidth}
\begin{displaymath}
\xymatrix{
    a_1 \ar[r]^b\ar@/^0.5pc/[d]^a   &   R   \\
    a_0 \ar[r]^b\ar@/^0.5pc/[u]^a   &   b_1 \ar@/^0.5pc/[r]^b\ar[u]^a &
    b_0\ar@/^0.5pc/[l]^b\ar@/_1pc/[ul]_a
}
\end{displaymath}
\end{minipage}
\begin{minipage}{0.5\textwidth}
\begin{displaymath}
\xymatrix{
    a_1 \ar[r]^b\ar@/^0.5pc/[d]^a   &   R_1 \ar@/^0.5pc/[r]^b   &   R_0\ar@/^0.5pc/[l]^b    \\
    a_0 \ar[r]^b\ar@/^0.5pc/[u]^a   &   b_1 \ar@/^0.5pc/[r]^b\ar[u]^a &
    b_0\ar@/^0.5pc/[l]^b\ar[u]^a
}
\end{displaymath}
\end{minipage}
\caption{Transition functions of $A_{\rm{min}}$ and $A'$.}\label{fig:example}
\end{figure}
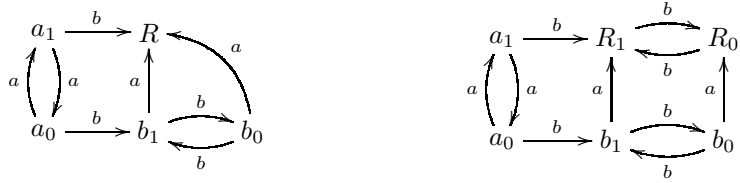

Now let us examine the automaton
$A'=(K',\Sigma_L,\delta',a_0,\{a_0,b_0\})$ with the transition function $\delta'$
defined by the second transition diagram in Fig.\ref{fig:example}.
Clearly, $L(A')=L(A_{\rm{min}})$, but by inspecting the lattice of S.P. partitions of $A'$, we
can find the pair $\pi_1=\{\{a_0\},\{a_1\},\{b_0,b_1\},\{R_0,R_1\}\}$ and
$\pi_2=\{\{a_0,a_1,b_0,R_0\},\{b_1,R_1\}\}$ such that $\pi_1\cdot\pi_2=0$ and they separate the
final states of $A'$. By Theorem~\ref{thm:SB-condition} we can use these partitions to
construct a nontrivial ASB- (and thus also SB-) decomposition of $A'$ formed by the automata $A_1$ and $A_2$
having two and four states, respectively.
Note that both $A_1$
and $A_2$ have less states than $A_{\rm{min}}$.
\end{Ex}

In the following theorem (inspired by a similar theorem in \cite{HartmanisStearns}) we
state a condition, under which the situation from the last example cannot
occur, i.e., under which any SB-decomposition of a DFA implies a (maybe simpler)
SB-decomposition of the equivalent minimal DFA. 

\begin{Thm}\label{thm:distributive}
Let $A=(K,\Sigma,\delta,q_0,F)$ be a deterministic finite automaton and let $A_{\rm{min}}=(K_{\rm{min}},\Sigma,\delta_{\rm{min}},q_{\rm{min}},F_{\rm{min}})$ be the minimal DFA such that $L(A)=L(A_{\rm{min}})$. Let
$(A_1,A_2)$ be a nontrivial SB-decomposition of $A$ consisting of automata having $k$ and $l$ states. If the lattice
of S.P. partitions of $A$ is distributive, then there exists an SB-decomposition of
$A_{\rm{min}}$ consisting of automata having $k'$ and $l'$ states, such that $k'\leq k$ and $l'\leq l$.
\end{Thm}

\begin{proof}
Since $A_{\rm{min}}$ is the minimal DFA such that $L(A)=L(A_{\rm{min}})$, there exists a
mapping $f\colon K\to K_{\rm{min}}$ such that
$(\forall w\in\Sigma^*);~f(\delta(q_0,w))=\delta_{\rm{min}}(q_{\rm{min}},w)$.
Using the mapping $f$, let us define a partition $\rho$ on the set of states of $A$ by
$p\equiv_{\rho}q \Leftrightarrow f(p)=f(q)$. Clearly, $\rho$ is
an S.P. partition.

Since $(A_1,A_2)$ is a nontrivial SB-decomposition of $A$, we can use it to obtain S.P. partitions
$\pi_1$ and $\pi_2$ on the set of states of $A$ such that $\pi_1\cdot\pi_2=0$. Let us
define new partitions $\pi_1'$ and $\pi_2'$ on the set of states of $A_{\rm{min}}$ by
$f(p)\equiv_{\pi_i'}f(q)\Leftrightarrow p\equiv_{\rho+\pi_i} q$. Since it
holds that $\rho+\pi_i\preceq\rho$, this definition does not depend on the
choice of the states $p$ and $q$. It holds that $|\pi_i'|=|\rho+\pi_i|\leq|\pi_i|$,
therefore if we prove that $\pi_1'$ and $\pi_2'$ are S.P. partitions and
$\pi_1'\cdot\pi_2'=0$, we can use them to construct the desired decomposition.

The fact that $\pi_i'$ is an S.P. partition on the set of states of $A_{\rm{min}}$ is a trivial
consequence of the fact that $\rho+\pi_i$ is an S.P. partition on the set of states of
$A$. We need to prove that $\pi_1'\cdot\pi_2'=0$. Let us assume that $p'$ and $q'$ are
states of $A_{\rm{min}}$ such that $p'\equiv_{\pi_1'\cdot\pi_2'}q'$ and $p,q$ are some states of $A$
such that $f(p)=p'$ and $f(q)=q'$. Then $p'\equiv_{\pi_1'}q'$ and
$p'\equiv_{\pi_2'}q'$, and by definition of $\pi_i'$ we get
$p\equiv_{\rho+\pi_1}q$ and $p\equiv_{\rho+\pi_2}q$, which is equivalent to
$p\equiv_{(\rho+\pi_1)\cdot(\rho+\pi_2)}q$. Since the lattice of all S.P. partitions of
$A$ is distributive, we have
$(\rho+\pi_1)\cdot(\rho+\pi_2)=\rho+(\pi_1\cdot\pi_2)=\rho+0=\rho$,
therefore $p\equiv_{\rho}q$, which by definition of $\rho$ implies that $f(p)=f(q)$, in
other words $p'=q'$. Hence $\pi_1'\cdot\pi_2'=0$.
\end{proof}

\section{Degrees of Decomposability}\label{sec:degrees}

It is easy to see that for each type of decomposition, there exist undecomposable
regular languages (e.g. $L^{(n)}=\{a^k | k\geq n-1\}$ is wAI-undecomposable for each
$n\in\mathbb{N}$). There also
exist regular languages, that are perfectly decomposable in
each way (e.g. $L^{(k,l)}=\{w\in\{a,b\}^* | \#_a(w)\textrm{~mod~} k=0 \land \#_b(w)\textrm{~mod~}
l=0\}$ has a perfect ASB-decomposition for all $k,l\geq 2$). We shall now investigate whether all values between these two
limits can be achieved.

\begin{Def}
Let $A$ be a DFA, let $(A_1,A_2)$ be its nontrivial SB- (ASB-) decomposition with the
corresponding S.P. partitions $\pi_1$ and $\pi_2$. We shall call this decomposition
\emph{redundant}, if there exist S.P. partitions $\pi_1'\succeq \pi_1$ and $\pi_2'\succeq
\pi_2$ such that at least one of these inequalities is strict, but it still holds
$\pi_1'\cdot\pi_2'=0$ (and $\pi_1'$ and $\pi_2'$ separate the final states of $A$).
\end{Def}

\begin{Lemma}\label{lemma:rectangle}
For each $r,s\in\mathbb{N}$, $r,s\geq 2$,  there exists a minimal DFA $A$ consisting of $r.s$ states
and having only one nontrivial nonredundant SB-decomposition (ASB-decomposition) up to the
order of automata, consisting of automata having $r$ and $s$ states.
\end{Lemma}

\begin{proof}
Let us study the minimal automaton $A_{r,s}=(K,\Sigma,\delta,q_{0,0},F)$ defined by
$K=\{q_{i,j}|i\in\{0,\ldots,r-1\},j\in\{0,\ldots,s-1\}\}$, $F=\{q_{r-1,s-1}\}$ and the
transition function $\delta$ defined by
\begin{eqnarray*}
\delta(q_{i,j},a)&=&q_{i+1,j} \rm{~for~} i\in\{0, \ldots,r-2\},j\in\{0, \ldots,s-1\}    \\
\delta(q_{r-1,j},a)&=&q_{r-1,j} \rm{~for~} j\in\{0, \ldots,s-1\}    \\
\delta(q_{i,j},b)&=&q_{i,j+1} \rm{~for~}  i\in\{0, \ldots,r-1\},j\in\{0, \ldots,s-2\}   \\
\delta(q_{i,s-1},b)&=&q_{i,s-1} \rm{~for~} i\in\{0, \ldots,r-1\}.
\end{eqnarray*}
To inspect the SB-decompositions of $A_{r,s}$, let us study the S.P. partitions on the set of
its states. From the method for generating all S.P. partitions of an automaton that is described
in \cite{HartmanisStearns}, we know that each nontrivial S.P. partition can be obtained as
a sum of some partitions $\pi_{p,t}^m$, where $\pi_{p,t}^m$ denotes the minimal S.P.
partition such that it does not distinguish between states $p$ and $t$, i.e., they belong
into the same block. Let us determine $\pi_{p,t}^m$ for various states $p$ and $t$ of $A_{r,s}$.

First, let us consider the case of $\pi_{p,t}^m$ such that $p=q_{i,j}$, $t=q_{i',j'}$ and
both inequalities $i<i'$ and $j<j'$ hold. Since $q_{i,j}\equiv_{\pi} q_{i',j'}$,
$\delta(q_{i,j},a^{i'-i}b^{j'-j})=q_{i',j'}$ and
$\delta(q_{i',j'},a^{i'-i}b^{j'-j})=q_{2i'-i,2j'-j}$ (if $2i'-i<r$ and $2j'-j<s$), as a consequence of the
substitution property of $\pi$, we obtain $q_{i,j}\equiv_{\pi} q_{2i'-i,2j'-j}$. By
applying this argument a finite number of times (keeping in mind the construction of $A_{r,s}$),
we obtain $q_{i,j}\equiv_{\pi} q_{r-1,s-1}$. Now let $k\in\{i,\ldots,r-1\}$ and let
$l\in\{i,\ldots,s-1\}$. Then $\delta(q_{i,j},a^{k-i}b^{l-j})=q_{k,l}$ and
$\delta(q_{i',j'},a^{k-i}b^{l-j})=q_{k+i'-i,l+j'-j}$ (if such states exist),
therefore $q_{k,l}\equiv_{\pi} q_{k+i'-i,l+j'-j}$. Again, we can use the same argument to
show that $q_{k,l}\equiv_{\pi} q_{r-1,s-1}$. Therefore, for this type of $\pi=\pi_{p,t}^m$, we
have $q_{k,l}\equiv_{\pi} q_{k',l'}$ for all $k,l,k',l'$ such that $i\leq k,k'<r$ and
$j\leq l,l'<s$.

Now let us consider the case of $\pi_{p,t}^m$ such that $p=q_{i,j}$, $t=q_{i',j'}$ and
it holds $i>i'$ and $j<j'$. Since $q_{i,j}\equiv_{\pi} q_{i',j'}$,
$\delta(q_{i,j},a^{r-1-i}b^{s-1-j'})=q_{r-1,s-1-(j'-j)}$ and
$\delta(q_{i',j'},a^{r-1-i}b^{s-1-j'})=q_{r-1-(i-i'),s-1}$, as a consequence of
the substitution property of $\pi$, we have $q_{r-1,s-1-(j'-j)}\equiv_{\pi}
q_{r-1-(i-i'),s-1}$. By exploiting the substitution property again on this equivalence,
using the words $a^{i-i'-1}$, $b^{j'-j-1}$ and $b^{j'-j}$, we obtain
$q_{r-2,s-1}\equiv_{\pi}q_{r-1,s-1}\equiv_{\pi}q_{r-2,s-2}$. Therefore in this case, no
such  $\pi_{p,t}^m$ partition can distinguish between states $q_{r-2,s-1}$, $q_{r-1,s-1}$
and $q_{r-2,s-2}$.

The last case to consider is the case of $\pi_{p,t}^m$ such that $p=q_{i,j}$,
$t=q_{i',j'}$ and it holds $i=i'$ (the case $j=j'$ is analogous). Without loss of
generality, we can assume that $j<j'$. Now, using the same arguments as in the first case,
we can show that $q_{i,l}\equiv_{\pi}q_{i,l'}$ for all $l,l'$ such that $j\leq
l,l'<s$. Therefore for each given $k$ such that $i\leq k<r$, it holds that
$q_{k,l}\equiv_{\pi}q_{k,l'}$ and all of the states not mentioned in this equivalence form
separate blocks of $\pi_{p,t}^m$.

It is easy to verify that one nontrivial ASB-decomposition of $A_{r,s}$ is given by the S.P. partitions
\begin{eqnarray*}
\pi_1&=&\{\{q_{0,0},\ldots,q_{0,s-1}\},\{q_{1,0},\ldots,q_{1,s-1}\},\ldots,\{q_{r-1,0},\ldots,q_{r-1,s-1}\}\}
\enspace\rm{and}\\
\pi_2&=&\{\{q_{0,0},\ldots,q_{r-1,0}\},\{q_{0,1},\ldots,q_{r-1,1}\},\ldots,\{q_{0,s-1},\ldots,q_{r-1,s-1}\}\}
\end{eqnarray*}
Now we show that any other SB-decomposition of $A_{r,s}$ is given by S.P. partitions preceding
to $\pi_1$ and $\pi_2$ in the partial order $\preceq$ and therefore is redundant.

Indeed, notice that none of the $\pi_{p,t}^m$ partitions of the first and the second
discussed type can distinguish between any of the states $q_{r-2,s-1}$, $q_{r-1,s-1}$
and $q_{r-2,s-2}$, therefore no sum of them can, either. For the partitions of the third
type, it holds either $q_{r-2,s-1}\equiv_{\pi}q_{r-1,s-1}$ or
$q_{r-1,s-1}\equiv_{\pi}q_{r-2,s-2}$, therefore it will take two partitions to distinguish
between these three states. Hence any nontrivial SB-decomposition is determined by two
S.P. partitions, both of which must be of the third type. But it is easy to see that for
any partition $\pi$ of this type it holds either $\pi\preceq\pi_1$ or $\pi\preceq\pi_2$.
\end{proof}

\begin{Def}\label{def:extension}
Let $A=(K,\Sigma,\delta,q_0,F)$ be a deterministic finite automaton, let
$K\cap\{p_0,p_1,\ldots,p_{k-1}\}=\emptyset$ and let $c$
be a new symbol not included in $\Sigma$. We shall define a \emph{$k$-extension} $A'$ of the automaton $A$ by the
following construction: $A'=(K\cup\{p_0,p_1,\ldots,p_{k-1}\},\Sigma\cup\{c\},\delta',p_0,F)$,
where the transition function $\delta'$ is defined as follows:
\begin{eqnarray*}
(\forall q\in K)~(\forall a\in\Sigma);~~~~\delta'(q,a)&=&\delta(q,a)\\
(\forall q\in K);~~~~                    \delta'(q,c)&=&q\\
(\forall p\in \{p_0,p_1,\ldots,p_{k-1}\})~(\forall a\in\Sigma);~~~~\delta'(p,a)&=&p\\
(\forall i\in \{0,1,\ldots,k-2\});~~~~                        \delta'(p_i,c)&=&p_{i+1}\\
\delta'(p_{k-1},c)&=&q_0 .
\end{eqnarray*}
\end{Def}


\begin{Lemma}\label{lemma:extension}
Let $A$ be a DFA consisting of $n$ states, all of which are reachable. Let $A'$ be its
$k$-extension. Then $A$ has a nontrivial nonredundant SB-decomposition (ASB-decomposition)
consisting of automata having $r$ and $s$ states iff
$A'$ has a nontrivial nonredundant decomposition of the same type, consisting of automata
having $k+r$ and $k+s$ states.
\end{Lemma}

\begin{proof}
We will try to inspect S.P. partitions on the set of states of $A'$, using the notation
from Definition~\ref{def:extension}. Let us assume that $\pi'$ is an S.P. partition on the
set of states of $A'$ such
that $p_i$ and $p_j$ are in the same block of $\pi'$; \hbox{$i,j\in\{0,1,\ldots,k-1\}$}.
As a consequence of the S.P. property, if $i,j<k-1$ then also $p_{i+1}$ and $p_{j+1}$ are in
the same block of $\pi'$, because $\delta'(p_i,c)=p_{i+1}$ and $\delta'(p_j,c)=p_{j+1}$. By
applying this argument a finite number of times, we can show that there exists some $l\in
\{0,1,\ldots,k-2\}$ such that $p_l\equiv_{\pi'} p_{k-1}$, and using the argument once more,
we obtain $p_{l+1}\equiv_{\pi'} q_{0}$. However, it holds $\delta'(p_l,a)=p_l$ for all
$a\in\Sigma$, hence $p_l\equiv_{\pi'} \delta'(q_0,w)$ for all $w\in\Sigma^*$. Since
all of the states of $A'$ are reachable, we have $p_l\equiv_{\pi'} q$ for all $q\in K$. Thus
such a partition cannot distinguish between the original states of the automaton $A$.

Now let us suppose that $\pi'$ is an S.P. partition on the set of states of $A'$ such
that for some \hbox{$i\in\{0,1,\ldots,k-1\}$}, $p_i\equiv_{\pi'} q$ for some $q$ in $K$.
Then it also holds that $p_i\equiv_{\pi'} p_{i+1}$, because
$\delta(p_i,c)=p_{i+1}$ and $\delta(q,c)=q$. But we have
already shown that $p_i\equiv_{\pi'} p_{i+1}$ implies
that all of the states in $K$ are equivalent modulo $\pi'$, thus this S.P. partition cannot
distinguish between the states of $A$, either.

>From these observations it follows that if $\pi'$ is any S.P. partition on the set of
states of $A'$ such
that the states of $A$ are not all equivalent modulo $\pi'$, then $\pi'$
must also contain $k$ blocks, each of which contains only one state
$p_i$, where $i\in\{0,1,\ldots,k-1\}$. Now we can prove the equivalence stated in the theorem.

Let $A$ have an SB-decomposition consisting of $r$ and $s$ states. Then there exist S.P.
partitions $\pi_1$ and $\pi_2$ on the set of states of $A$ having $r$ and $s$ blocks, such that $\pi_1\cdot\pi_2=0$.
Let us now construct new partitions $\pi_1'$ and $\pi_2'$ on the set of states of $A'$ by
$\pi_1'=\pi_1\cup\left\{\{p_0\},\{p_1\},\ldots,\{p_{k-1}\}\right\}$ and
$\pi_2'=\pi_2\cup\left\{\{p_0\},\{p_1\},\ldots,\{p_{k-1}\}\right\}$.
Obviously, $\pi_1'$ and $\pi_2'$ have substitution property, because for the states in $K$
this property is
inherited from $\pi_1$ and $\pi_2$, and the new states $p_0, p_1, \ldots, p_{k-1}$ cannot
violate this property either, because each of these states belongs to a separate block in $\pi_1'$
and $\pi_2'$, making the substitution property hold trivially. Neither do the
new $c$-moves defined on the states from $K$ violate the substitution property. Finally, it holds that
\mbox{$\pi_1'\cdot\pi_2'=0$}. To see this, note that for a state $q\in K$, it holds
$[q]_{\pi_1'\cdot\pi_2'}=[q]_{\pi_1\cdot\pi_2}=\{q\}$, since $\pi_1\cdot\pi_2=0$. For a
state $q\in K' - K$, $[q]_{\pi_i'}=\{q\}$ for $i\in\{1,2\}$ thus
$[q]_{\pi_1'\cdot\pi_2'}=\{q\}$, too. Hence each state of $A'$ belongs to a separate block
of $\pi_1'\cdot\pi_2'$, which implies $\pi_1'\cdot\pi_2'=0$. Therefore $\pi_1'$ and
$\pi_2'$ induce an SB-decomposition of $A'$.
It is also easy to see that if $\pi_1$ and $\pi_2$ separate the final states of $A$, then
also $\pi_1'$ and $\pi_2'$ separate the final states of $A'$, making the induced
decomposition an ASB-decomposition.

On the other hand, let us now assume that $A'$ has an SB-decomposition and $\pi_1'$ and $\pi_2'$ are the S.P.
partitions on $K'$ that induce this decomposition, thus \mbox{$\pi_1'\cdot\pi_2'=0$}. From the
observations made in the beginning of this proof, we know that any S.P. partition that
can distinguish between the states in $K$ in any way, must contain each of the states $p_0,
p_1\ldots p_{k-1}$ in a separate block containing only this state. As
$\pi_1'\cdot\pi_2'=0$, for all $q_1,q_2\in K$, at least one of these partitions must
distinguish between these states, i.e., $[q_1]_{\pi_i'}\not =[q_2]_{\pi_i'}$. If one of the
partitions distinguished between all such pairs, it would imply that this partition must
contain a separate block for each one of the states in $K'$, thus becoming a trivial
partition $0$, resulting in a trivial decomposition. Therefore both $\pi_1'$ and
$\pi_2'$ have to distinguish between some pair of states from $K$, which implies that they
both contain a separate block for each of the states $p_0,p_1\ldots p_{k-1}$ containing no
other state. By removing these $k$ blocks from $\pi_1'$ and $\pi_2'$, we obtain new
partitions $\pi_1$ and $\pi_2$ on
the set $K$, such that
$\pi_1=\pi_1' - \left\{\{p_0\},\{p_1\},\ldots,\{p_{k-1}\}\right\}$ and
$\pi_2=\pi_2' - \left\{\{p_0\},\{p_1\},\ldots,\{p_{k-1}\}\right\}$.
These partitions preserve the substitution property, since $(\forall a\in\Sigma)(\forall q\in
K)$: $\delta(q,a)\in K$ and $\pi_1'$ and $\pi_2'$ were S.P. partitions. It also holds
$\pi_1\cdot\pi_2=0$, as for all $q_1,q_2\in K$, $q_1\equiv_{\pi_1\cdot\pi_2} q_2$ implies
$q_1\equiv_{\pi_1'\cdot\pi_2'} q_2$ and that implies $q_1=q_2$. So $\pi_1$ and $\pi_2$
induce an SB-decomposition of $A$. As $\pi_1'$ and $\pi_2'$ were nontrivial, so are
$\pi_1$ and $\pi_2$ and the obtained decomposition. It is again easy to see that if
$\pi_1'$ and $\pi_2'$ separate the final states of $A'$, then also $\pi_1$ and $\pi_2$
must separate the final states of $A$.

The described relationship between the S.P. partitions on the set of states of $A$ and the
corresponding S.P. partitions on $A'$ also implies, that each decomposition of A is
nonredundant iff the corresponding decomposition of $A'$ is nonredundant, too.
\end{proof}

Since a $k$-extension of a minimal DFA is again a minimal DFA, we can combine
the lemmas to obtain the following theorem.

\begin{Thm}
Let $n\in\mathbb{N}$ be such that $n=k+r.s$, where $r,s, k\in\mathbb{N}$, $r,s\geq 2$. Then there exists a
minimal DFA $A$ consisting of $n$ states, such
that it has only one nontrivial nonredundant SB-decomposition (ASB-decomposition) up to
the order of the automata in the decomposition, and this decomposition consists of
automata with $k+r$ and $k+s$ states.
\end{Thm}



\end{document}